\documentclass[conference]{IEEEtrans}
\usepackage{amsmath}
\usepackage{amssymb}
\usepackage{amsfonts}
\usepackage{graphicx}
\usepackage{float}
\usepackage{graphics}
\usepackage{graphicx}
\usepackage{theorem}
\usepackage{euscript}
\usepackage{psfrag}

\newcommand{\Mat}[1]{\boldsymbol{#1}}

\DeclareMathOperator{\diag}{diag}

\def\E{{\mathbb E}}

\newcommand{\myscale}{0.348}

\theoremstyle{plain}
\newtheorem{thm}{Theorem}

\newtheorem{prop}[thm]{Proposition}

\begin{document}

\title{Throughput of Cellular Uplink with Dynamic User Activity and Cooperative Base-Stations}
\author{%
\authorblockN{O. Somekh\authorrefmark{1}, O.
Simeone\authorrefmark{2}, H. V. Poor\authorrefmark{1}, and S. Shamai
(Shitz)\authorrefmark{3}} \\
\authorblockA{\authorrefmark{1}
Department of Electrical Engineering, Princeton University, Princeton, NJ 08544, USA, {\{orens, poor\}@princeton.edu}}
\authorblockA{\authorrefmark{2}
CWCSPR, New Jersey Institute of Technology, Newark, NJ 07102-1982, USA, {osvaldo.simeone@njit.edu.}}
\authorblockA{\authorrefmark{3}
Department of Electrical Engineering,Technion, Haifa, 32000, Israel, {\{sshlomo\}@ee.technion.ac.il}}
}
\maketitle

\begin{abstract}
The throughput of a linear cellular uplink with a random number of users,
different power control schemes, and cooperative base stations is considered
in the large system limit where the number of cells is large for non fading
Gaussian channels. The analysis is facilitated by establishing an analogy
between the cellular channel per-cell throughput with joint multi-cell
processing (MCP), and the rate of a deterministic inter-symbol interference
(ISI) channel with flat fading. It is shown that, under certain conditions,
the dynamics of cellular systems (i.e., a random number of users coupled with 
a given power control scheme) can be interpreted, as far as the uplink
throughput is concerned, as the flat fading process of the equivalent ISI
channel. The results are used to demonstrate the benefits of MCP over the
conventional single cell processing approach as a function of various system
parameters in the presence of random user activity.
\end{abstract}



\vspace{-0.0cm}

\section{Introduction}

Wireless communications systems in general and cellular systems in
particular are of major interest as they allow the provision of continuous services to mobile users. In recent years a considerable research effort has been devoted to the development of new technologies for providing better services
and extending system coverage. In this context, the use of joint multi-cell
processing (MCP) has been identified as a key tool for enhancing system
performance. Since its introduction by Wyner in \cite{Wyner-94}, many
aspects of MCP has been studied (see 
\cite{Somekh-Simeone-Barness-Haimovich-Shamai-BookChapt-07} and references
therein for a survey of recent results on MCP). Here we are interested in
studying the impact of users with dynamic activity (i.e., each user is
active with a certain probability in each time slot) on the performance of
cellular uplink with MCP. Early attempts to deal with random number of users
in cellular systems focused on single cell processing (SCP) (e.g., \cite%
{Shamai-Wyner-97-II}), and were based on the notion of a random
multiple-access channel (MAC) \cite{Plotnik-ISIT90}. In a recent work \cite%
{Somekh-Simeone-Poor-Shamai-Allerton08}, the per-cell throughput of a simple
infinite linear cellular uplink with a \textit{single} dynamic user per cell
is analyzed. The analysis relies on the special topology of the model in which
interference stems from a single neighboring cell only. In a parallel work 
\cite{Levy-Zeitouni-Shamai-2Tap_UP08}, the authors use similar tools to
consider also the resulting rate statistics to derive outage performance for
the same cellular uplink.

In this work, we extend the results of \cite%
{Somekh-Simeone-Poor-Shamai-Allerton08}, derived for a single user per cell,
and study the case of more than one dynamic user per cell in the
large-system limit. In particular, we calculate the per-cell throughput
supported by a simple linear infinite non-fading cellular uplink model,
in which the number of users in each cell is a binomially distributed random
variable (r.v.), and all the BSs jointly decode their received signals to
recover the users' messages. To facilitate analytical treatment we use a
linear variant of the Wyner cellular model family where each user
\textquotedblleft sees" only a finite (but arbitrary) number of BSs \cite%
{Wyner-94}. The main analytical tool used here is a recent result by Tulino 
\textit{et al.} that provides expressions for the achievable rates of a
linear time invariant (LTI) inter-symbol interference (ISI) channel with
flat fading applied to its output symbols \cite%
{Tulino-Verdu-Caire-Shamai-ISIT08}. By establishing an analogy between the
per-cell throughput of the cellular uplink and the rate of an ISI channels
(similarly to Wyner \cite{Wyner-94}), we show that results of \cite%
{Tulino-Verdu-Caire-Shamai-ISIT08} can be applied to the cellular setup to
address the dynamic setting at hand. In particular, the path gain between a
user and the BSs are interpreted as the ISI channel coefficients and the
cellular power control scheme determines the fading statistics of the
equivalent ISI channel. We use the results to demonstrate the benefits of
MCP over SCP for a cellular system with dynamic user activity. In a
related work \cite{Tulino-Verdu-Caire-Shamai_ISIT07}, achievable rates for
an output-erasure ISI channel were derived and used to calculate the
per-cell throughput of a cellular uplink with MCP and base-stations
subjected to backhaul failures.

\section{System Model}

\begin{figure}[t]
\begin{center}
\includegraphics[angle=-90, scale=0.305]{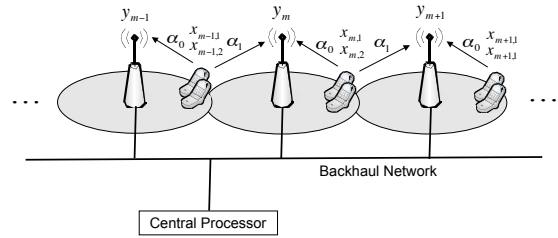} \vspace{-3.0cm}
\end{center}
\caption{The soft-handoff cellular uplink for $K=2$.}
\label{fig: SHO model}
\end{figure}

We consider a linear cellular Gaussian (no fading) uplink, where $M$
identical cells are arranged along a line \cite{Wyner-94}. Each cell includes a
BS and $K$ identical mobile terminals (MTs). While ignoring boundary
effects, it is assumed that each MT's transmissions are received by $%
L_{1}+L_{2}+1$ BSs only: its local BS, the $L_{1}$ adjacent BSs on its
left, and the $L_{2}$ adjacent BSs on its right (See Fig. \ref{fig: SHO
model} for the special case of $L_{1}=0$, $L_{2}=1$, and $K=2$). In
particular, the signals of the $m_{1}$th cell are received at the $m_{2}$th BS
with signal path gain $\alpha _{m_{1}-m_{2}}\in \mathbb{C}$. We further
assume perfect symbol and block synchronization, and that the total cell
transmission power is subjected to an average power constraint $P$. As
mentioned above we consider a dynamic model in which during each slot (or
transmission block) MTs are randomly and independently active with
probability $(1-q)$, and otherwise are kept silent by means of control or lack of input
data. Under these assumptions, the received signal at the $m$th BS
for an arbitrary symbol of an arbitrary block is 
\begin{equation}
y_{m}=\sum_{l=-L}^{L}\sum_{k=1}^{K}x_{m+l,k}e_{m+l,k}\alpha _{l}+z_{m}\ ;%
\begin{array}{c}
m=1,\ldots ,M \\ 
k=1,\ldots ,K%
\end{array}%
,  \label{eq: received signal}
\end{equation}%
where $x_{n,k}$ is the $n$th cell $k$th MT transmission, $x_{n,k}\sim 
\mathcal{CN}(0,p_{n,k})$, $e_{n,k}\in \lbrack 0,1]$ is the corresponding
i.i.d. \textit{Bernoulli} activity r.v., $e_{n,k}\sim \mathcal{B}(1-q)$, $%
z_{m}$ is the additive Gaussian noise, $z_{m}\sim \mathcal{CN}(0,1)$, and
out-of-range indices should be ignored. The transmission powers $\{p_{m,k}\}$
are functions of the activity pattern and the selected power control method
(to be discussed later on). The activity r.v.'s are assumed to be independent and identically distributed (i.i.d.)
among the users but are non-ergodic along the time index for each user.
Finally, in order to satisfy a per-cell power constraint any power
allocation must satisfy 
\begin{equation}
\sum_{k=1}^{K}p_{m,k}\leq P\quad ,\quad \forall m\ .
\end{equation}

\section{Preliminary}

\vspace{-0.0cm}

The main analytical tool we use in this work is reported in a recent work 
\cite{Tulino-Verdu-Caire-Shamai-ISIT08}, in which Tulino \textit{et al.}
study the capacity of a deterministic inter-symbol interference (ISI) flat
fading channel (depicted in Fig. \ref{fig: ISI flat fading channel}). The
channel includes a unit power stationary Gaussian input $x_{i}$, with power
spectral density (PSD) $S_{x}(f)$, which enters a linear time invariant
filter $H(f)$. The output of the latter is then multiplied by a flat fading
i.i.d. process $\sqrt{\gamma }A_{i}$ where $\gamma $ is a non-negative
constant, and corrupted by zero mean unit power white Gaussian noise $z_{i}$. 
Assuming that only the decoder is aware of the filter coefficient, the
fading process, the constant $\gamma $ and the statistics of the input and
noise signals, the \textit{ergodic} input-output mutual information is
proved in \cite{Tulino-Verdu-Caire-Shamai-ISIT08}\footnote{%
It is noted that this result is actually taken from the presentation slides
and is more compact than that reported in the conference proceedings \cite%
{Tulino-Verdu-Caire-Shamai-ISIT08}.} to be given by 
\begin{multline}
I(\gamma )=\int_{0}^{1}\log _{2}\left( 1+\gamma \beta S(f)\right) df+
\label{eq: Tulino rate} \\
+\E\left( \log _{2}\left( 1+\gamma \nu \left\vert A\right\vert ^{2}\right) %
\right) -\log _{2}\left( 1+\gamma \beta \nu \right) \ ,
\end{multline}%
where $S(f)=S_{x}(f)\left\vert H(f)\right\vert ^{2}$ is the filter output
PSD, and $\beta$ and $\nu $ are the unique positive solutions to 
\begin{equation}
\E\left( \frac{1}{1+\gamma \nu \left\vert A\right\vert ^{2}}\right) =\frac{1%
}{1+\gamma \beta \nu }=\int_{0}^{1}\frac{1}{1+\gamma \beta S(f)}df\ .
\label{eq: Tulino fix point}
\end{equation}%
In the special case in which $A\in \{0,1\}\sim \mathcal{B}(1-\tilde{q})$ (an
output erasure channel \cite{Somekh-Simeone-Poor-Shamai-Allerton08}\cite{Tulino-Verdu-Caire-Shamai_ISIT07}), the mutual information 
\eqref{eq:
Tulino rate} reduces to 
\begin{equation}
I(\gamma )=\int_{0}^{1}\log _{2}\left( 1+\gamma \beta S(f)\right) df+d(%
\tilde{q}\Vert 1-\beta )\ ,  \label{eq: Tulino rate erasure}
\end{equation}%
where $d(\tilde{q}\Vert 1-\beta )$ is the relative entropy (in bits) between 
$\mathcal{B}(1-\tilde{q})$ and $\mathcal{B}(\beta )$, and $0\leq \beta \leq
1-\tilde{q}$ is the unique solution to 
\begin{equation}
\frac{\tilde{q}}{1-\beta }=\int_{0}^{1}\frac{1}{1+\gamma \beta S(f)}df\ .
\label{eq: Tulino fix point erasure}
\end{equation}


\begin{figure}[t]
\begin{center}
\includegraphics[angle=-90, scale=0.3]{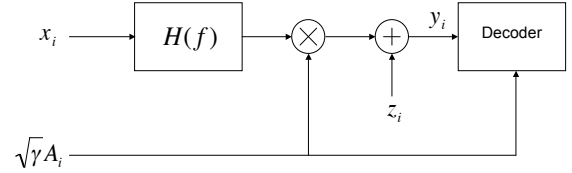} 
\vspace{-3.70cm}
\end{center}
\caption{The deterministic ISI flat fading channel studied in \protect\cite%
{Tulino-Verdu-Caire-Shamai-ISIT08}.}
\label{fig: ISI flat fading channel}
\end{figure}

\vspace{-0.0cm}

\section{Multicell Processing}

\vspace{-0.0cm}

With MCP, the BSs send their received signals to a central processor (CP)
via an ideal backhaul network. The CP collects the received signals and
jointly decodes the MTs' messages. We also assume that the CP is aware of
the activity pattern of all the MTs. Observing \eqref{eq: received signal},
the received signal vector at the BSs for an arbitrary symbol of an
arbitrary block is given by 
\begin{equation}
Y=\boldsymbol{H}\boldsymbol{E}X+Z\ ,  \label{eq: received signal vector}
\end{equation}%
where $X$ is the MTs' $MK\times 1$ transmit vector $X\sim \mathcal{CN}(%
\boldsymbol{0},\boldsymbol{Q})$ and $\boldsymbol{Q}=\diag{\{p_{1,1},%
\ldots,p_{M,K}\}}$,
$\boldsymbol{E}=\diag{\{e_{1,1}\ldots,e_{M,K}\}}$ is the diagonal $MK\times
MK$ activity matrix, $\boldsymbol{H}$ is the $M\times MK$, $L_{1}+L_{2}+1$
block diagonal channel transfer matrix $[\boldsymbol{H}]_{m,(n-1)K+k}=\alpha
_{n-m},\ k=1,2,\ldots \,K$, and $Z$ is the $M\times 1$ noise vector $Z\sim 
\mathcal{CN}(\boldsymbol{0},\boldsymbol{I})$.


With optimal MCP, the \emph{throughput} \emph{per cell} is a r.v., given in
the large system limit by 
\begin{equation}  \label{eq: MCP rate user activity}
R_{\mathrm{mcp}}= \lim_{M\rightarrow\infty} \frac{1}{M}\log_2 \det{\left(%
\boldsymbol{I}+ \boldsymbol{H}\boldsymbol{E}\boldsymbol{Q}\boldsymbol{E}%
^\dagger\boldsymbol{H}^\dagger\right)}\ .
\end{equation}
A close inspection of the covariance matrix $G=\boldsymbol{I}+ \boldsymbol{H}%
\boldsymbol{E}\boldsymbol{Q}\boldsymbol{E}^\dagger\boldsymbol{H}^\dagger$
reveals the following. 
\vspace{-0.2cm}

\begin{prop}
\label{Prop: MK to M MAC} With MCP the $MK$ user MAC of 
\eqref{eq: received
signal vector} is equivalent in terms of the throughput per cell to the
following $M$ user MAC 
\begin{equation}  \label{eq: received signal vector compact}
Y=\tilde{\boldsymbol{H}}\tilde{\boldsymbol{E}}\tilde{X}+Z\ ,
\end{equation}
where $\tilde{X}$ is an $M\times 1$ vector $\tilde{X}\sim\mathcal{CN}(%
\boldsymbol{0},\boldsymbol{I})$, $\tilde{\boldsymbol{E}}$ is a diagonal $%
M\times M$ matrix $[\tilde{\boldsymbol{E}}]_{m,m}=\sqrt{\sum_{k=1}^K
e_{m,k}p_{m,k}}$, and $\tilde{\boldsymbol{H}}$ is an $M\times M$ matrix $[%
\tilde{\boldsymbol{H}}]_{i,j}=\alpha_{j-i}$.
\end{prop}

\begin{proof}
To prove this claim it is enough to show that the covariance matrices of the
two Gaussian vectors conditioned on the activity r.v.'s (expressions %
\eqref{eq: received signal vector} and 
\eqref{eq: received signal vector
compact}) are equal. This is easily achieved by a straightforward matrix
multiplication, i.e., by showing that $\boldsymbol{H}\boldsymbol{E}%
\boldsymbol{Q}\boldsymbol{E}^\dagger\boldsymbol{H}^\dagger=\tilde{%
\boldsymbol{H}}\tilde{\boldsymbol{E}} \tilde{\boldsymbol{E}}^\dagger\tilde{%
\boldsymbol{H}}^\dagger$. (See also \cite[Lemma 3.1]{Wyner-94} for a similar
claim.)
\end{proof}

It is concluded that the system is equivalent in terms of throughput to a
system with a single user per-cell. Moreover, the transmissions of the $m$th
cell virtual user are multiplied by the r.v. $[\tilde{\boldsymbol{E}}]_{m,m}=%
\tilde{e}_{m,m}\triangleq\sqrt{\sum_{k=1}^K e_{m,k}p_{m,k}}$ which is a
function of the activity pattern and the power control scheme of the
original system.

Now, the ground is set for presenting our main observation. 
\vspace{-0.5cm}

\begin{prop}
\label{Prop: model equivalence} In the large system limit ($%
M\rightarrow\infty$) the per-cell throughput of the $MK$ user MAC of %
\eqref{eq: received signal vector}, is equal to the capacity of the
deterministic ISI channel with flat fading of \eqref{eq: Tulino rate}, with $%
S_x(f)=1$, $\gamma=1$, filter coefficients $\{\alpha_l\}_{l=-L_1}^{L_2}$
(i.e., $H(f)=\sum_{n=L_1}^{L_2} \alpha_n e^{-j2\pi n}$), and flat fading $A$
distributed as the i.i.d. diagonal elements of $\boldsymbol{\tilde{E}}$.
\end{prop}

\vspace{-0.0cm}

\begin{proof}
\textit{(outline)} Following similar argumentation as those made in \cite%
{Levy-Zeitouni-Shamai_CLT-UP08} we claim that in the large system limit the
per-cell throughput $R_{\mathrm{mcp}}$ (expression \eqref{eq: MCP rate user activity}) converges almost surely (a.s.) to its expected value 
\begin{equation}
R_{\mathrm{mcp}} = \lim_{M\rightarrow\infty} \E\left(\frac{1}{M}\log_2 \det{\left(\Mat{I}+\tilde{%
\boldsymbol{H}}\tilde{\boldsymbol{E}} \tilde{\boldsymbol{E}}^\dagger\tilde{%
\boldsymbol{H}}^\dagger\right)}\right)\ ,
\end{equation}
where the expectation is taken over the i.i.d. diagonal entries of $%
\boldsymbol{\tilde{E}}$. Next, we address the fact that the virtual flat
fading process $\boldsymbol{\tilde{E}}$ affects the MTs' signals $%
\boldsymbol{\tilde{X}}$ and not the outputs of the BSs $\boldsymbol{%
\tilde{H}}\boldsymbol{\tilde{X}}$. By recalling that 
\begin{equation}  \label{eq: input is output}
\det\left(\boldsymbol{I}+\tilde{\boldsymbol{H}}\tilde{\boldsymbol{E}} \tilde{%
\boldsymbol{E}}^\dagger\tilde{\boldsymbol{H}}^\dagger\right)= \det\left(%
\boldsymbol{I}+\tilde{\boldsymbol{E}}^\dagger\tilde{\boldsymbol{H}}^\dagger 
\tilde{\boldsymbol{H}}\tilde{\boldsymbol{E}}\right)\ ,
\end{equation}
it is evident that for i.i.d. input $\boldsymbol{\tilde{X}}$ the resulting
throughput is the same, regardless whether the virtual flat fading affects
the input signal or the output signal. Examining the right-hand-side (RHS)
of \eqref{eq: input is output} it is concluded that up to a transpose
conjugate of the inter-cell interference coefficient vector $\{\alpha_l\}$
the per-cell throughput of the cellular uplink and the rate of the ISI
channel are equivalent.

\end{proof}

It is noted that the last proposition holds for other user activity pattern
statistics. For example the number of users per-cell can be assumed to be unbounded and 
drawn according to a Poisson distribution (see \cite{Shamai-Wyner-97-II}). 
Next, we consider several possible power control policies that determine
the actual fading distribution of the equivalent ISI channel.

\subsection{No Power Control}

When no power control (NPC) is used, each \emph{active} user (i.e., $%
e_{m,k}=1$) is transmitting with a fixed power $p_{m,k}=P/K$. Hence, the
fact that the active users can increase their transmission power while the
cell still meets its total power constraint, is ignored. 
For NPC it is easily verified that $\{\tilde{e}_{m,m}\}_{m=1}^M$ are i.i.d.
r.v.'s $\tilde{e}_{m,m}=\sqrt{L_mP/K}$ where $\{L_m\}_{m=1}^M$ are i.i.d. Binomial r.v.'s $L_m\sim\mathcal{BN}%
(K,1-q)$. It can be shown that for large $K$ and fixed $P$ the virtual
fading process consolidates and the per-cell throughput converges (and is upper
bounded for any $K$, not necessarily large) to that of a static system ($K$
active users in each cell) but with power penalty of $P(1-q)$.

\subsection{Adaptive Power Control}

According to the adaptive power control (APC) scheme each active user in the 
$m$th cell transmits using power $p_{m,k}=P/K_m$, where $K_m\triangleq%
\sum_{k=1}^K e_{m,k}$ is the number of active users in the $m$th cell. In
this case it is easily verified that the total cell power constraint is
satisfied and that $\{\tilde{e}_{m,m}\}_{m=1}^M$ are i.i.d. r.v.'s 
$\tilde{e}_{m,m}=\sqrt{L_mP}$ where $\{L_m\}_{m=1}^M$ are i.i.d. Bernoulli r.v.'s $L_m\sim\mathcal{B}(1-q^K)$. As with NPC, it is
easily shown that for large $K$ and fixed $P$ the virtual fading process
consolidates and the per-cell throughput converges (and is upper bounded for
any $K$, not necessarily large) to that of a static system with \textit{no
power penalty}.

\subsection{Cognitive Power Control}

For the cognitive power control (CPC) policy we use the convention that
inactive users are assumed to be aware of all the active users' messages
(see \cite{Devroye-Vu-Tarokh-SPMag08}). Accordingly, each inactive user
divides its power and transmits the active users' messages in a coherent
manner. Straightforward calculations yield that the optimal power of the virtual
user is 
\begin{equation}  \label{ }
P_L^\mathrm{o} = \left\{%
\begin{array}{cc}
(K-L+1)P & 0< L\le K \\ 
0 & L=0%
\end{array}%
\right.\ .
\end{equation}
Hence, $\{\tilde{e}_{m,m}\}_{m=1}^M$ are i.i.d. r.v.'s  
$\tilde{e}_{m,m}=\sqrt{P^\mathrm{o}_{L_m}}$ where $\{L_m\}_{m=1}^M$ are i.i.d. Binomial r.v.'s $L_m\sim\mathcal{BN}(K,1-q)$.
It is easily verified that the per-cell throughput in this case is upper
bounded for any $K$ by that of a static system but with power gain of $P\left(1+Kq-(K+1)q^K\right)$%
.

\section{Single-Cell Processing}

\begin{figure*}[tbp]
\begin{equation}  \label{eq: 2tap with APC beta explicit}
\beta = \frac{q^K\sqrt{2P(\left\vert \alpha_0\right\vert^2+\left\vert
\alpha_1\right\vert^2)+P^2(\left\vert \alpha_0\right\vert^4+\left\vert
\alpha_1\right\vert^4)-2P^2\left\vert \alpha_0\right\vert^2\left\vert
\alpha_1\right\vert^2(1-2q^{2K})+1}-q^{2K}P(\left\vert
\alpha_0\right\vert^2+\left\vert \alpha_1\right\vert^2)-1}{%
q^{2K}P^2\left(\left\vert \alpha_0\right\vert^2-\left\vert
\alpha_1\right\vert^2\right)^2-1}
\end{equation}%
\par
\begin{center}
\line(1,0){500}
\end{center}
\end{figure*}

For comparison purposes we consider single cell processing (SCP) schemes.
Here, each BS is aware of the activity pattern of its cell's users and these
of the interfering cells' users only. On the other hand it is aware of the
codebooks of its cell's users only while it is oblivious of the codebooks of
the interfering cells' users. Hence, each BS treats the signals stemming
from the interfering cells as Gaussian noise (conditioned on the activity
pattern). Accordingly, the sum-rate of the $m$th cell is an r.v. given by 
\begin{equation}
R_{\mathrm{scp}}(m) = \log_2\left(1+\mathrm{SINR}(m)\right)\ ,
\end{equation}
where $\mathrm{SINR}(m)$ is the signal to interference plus noise ratio at
the $m$th BS. Hence, the per cell throughput of the system is given by 
\begin{equation}  \label{eq: SCP rate SLLN}
\begin{aligned} R_{\mathrm{scp}}&=
\lim_{M\rightarrow\infty}\frac{1}{M}\sum_{m=1}^M
R_{\mathrm{scp}}(m)={\mathbb
E}\left[\log_2\left(1+\mathrm{SINR}\right)\right]\ , \end{aligned}
\end{equation}
where 
\begin{equation}  \label{ }
\mathrm{SINR} = \frac{\left\vert \alpha_0\right\vert^2 e^2_0}{1+\sum_{\substack{l=-L_1\\ l\neq0}}^{L_2} \left\vert \alpha_l\right\vert^2 e^2_l}\ ,
\end{equation}
and the expectation is taken over the arbitrary i.i.d. r.v.'s $%
\{e_l\}_{l=-L_1}^{L_2}$ which are distributed as the i.i.d. diagonal entries
of $\tilde{\boldsymbol{E}}$ according to the specific power control scheme
being used (i.e., NPC, APC, or CPC). The second equality of 
\eqref{eq: SCP rate SLLN} holds almost surely (a.s.) and is achieved by the strong law of large
numbers (SLLN). The latter is applicable here since the interference in each
cell stems from no more than $L_1$ neighboring cells to the left and $%
L_2$ neighboring cells to the right. Thus, $\mathrm{SINR}(m_1)$ and $\mathrm{%
SINR}(m_2)$ are i.i.d. for $\left\vert m_1-m_2\right\vert>L_1+L_2+1$.



\vspace{-0.0cm}

\section{Soft Handoff Model}

\vspace{-0.0cm}
Here we focus on the simplest instance of the considered
model. According to the \textit{soft-handoff} (SHO) model, depicted in Fig. %
\ref{fig: SHO model}, inter-cell interference stems from one adjacent cell
only (see \cite{Somekh-Zaidel-Shamai-SH-IT-2007}). In this case $L_1=0$, $%
L_2=1$, and only $\alpha_0$ and $\alpha_1$ are non-zero. The power spectral
density is given by 
\begin{equation}  \label{eq: 2tap spectrum}
S(f)=(\left\vert \alpha_0\right\vert^2+\left\vert
\alpha_1\right\vert^2)+2\left\vert \alpha_0\right\vert\left\vert
\alpha_1\right\vert\cos(2\pi f+\phi)\ ,
\end{equation}
where $\phi=\angle\ (\alpha_1 \alpha^\dagger_0)$. The integrals on the RHS
of \eqref{eq: Tulino rate} and \eqref{eq: Tulino fix point} reduce for the SHO model to 
\begin{equation}  \label{eq: 2tap rate}
\int_0^1\log_2\left(1+\gamma\beta S(f)\right)df=\log_2\left(\frac{a+\sqrt{%
a^2-b^2}}{2}\right)\ ,
\end{equation}
and 
\begin{equation}  \label{eq: 2tap fix point}
\int_0^1\frac{1}{1+\gamma\beta S(f)}df=\frac{1}{\sqrt{a^2-b^2}}\ ,
\end{equation}
respectively, where 
\begin{equation}
\begin{aligned} 
a\triangleq 1+\gamma\beta(\left\vert
\alpha_0\right\vert^2+\left\vert \alpha_1\right\vert^2)\ {\rm and}\ b\triangleq
2\gamma\beta\left\vert \alpha_0\right\vert\left\vert \alpha_1\right\vert\ .
\end{aligned}
\end{equation}
Expressions \eqref{eq: 2tap rate} and \eqref{eq: 2tap fix point} hold for
all power control schemes (or any equivalent fading distribution).

In case where APC is applied, hence $A\in\{0,1\}\sim \mathcal{B}(1-q^K)$ and 
$\gamma=P$, the result can be expressed in closed form (using 
\eqref{eq:
Tulino rate erasure}, \eqref{eq: Tulino fix point erasure}, 
\eqref{eq: 2tap
rate}, and \eqref{eq: 2tap fix point}), and the MCP rate is given by
\begin{equation}
R_{\mathrm{mcp}} = \log_2\left(\frac{a+\sqrt{a^2-b^2}}{2}\right)+d(q^K\|1-%
\beta)\ ,
\end{equation}
where $\beta$ is given explicitly by \eqref{eq: 2tap with APC beta explicit}.

\section{Numerical Results}

\begin{figure}[t]
\begin{center}
\includegraphics[scale= \myscale]{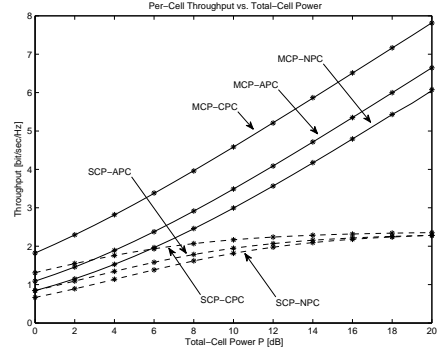} \vspace{-0.45cm}
\end{center}
\caption{Per-cell throughput vs. the total cell power $P$ supported by the
SHO model ($\protect\alpha_0=1,\ \protect\alpha_1=0.5$), for $K=5$ users
per-cell, and non-activity probability $q=0.3$.}
\label{fig: Rate vs power}
\end{figure}

\begin{figure}[t]
\begin{center}
\includegraphics[scale= \myscale]{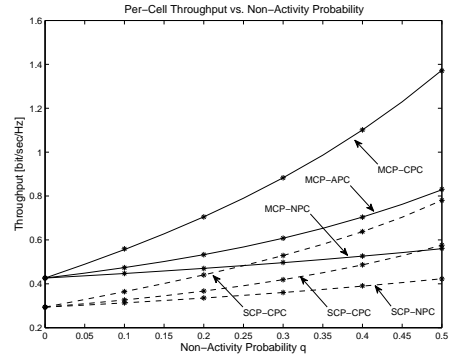} \vspace{-0.45cm}
\end{center}
\caption{Per-cell throughput vs. the non-activity probability $q$ supported
by the SHO model ($\protect\alpha_0=1,\ \protect\alpha_1=0.5$), for $K=5$
users per-cell, and total cell power $P=5\ [\mathrm{dB}]$.}
\label{fig: Rate vs q}
\end{figure}

\begin{figure}[t]
\begin{center}
\includegraphics[scale= \myscale]{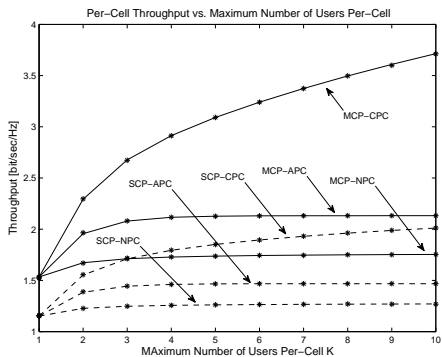} \vspace{-0.45cm}
\end{center}
\caption{Per-cell throughput vs. the number of users per-cell $K$ supported
by the SHO model ($\protect\alpha_0=1,\ \protect\alpha_1=0.5$), for
non-activity probability $q=0.3$, and total cell power $P=5\ [\mathrm{dB}]$%
.}
\label{fig: Rate vs K}
\end{figure}

\begin{figure}[t]
\begin{center}
\includegraphics[scale= \myscale]{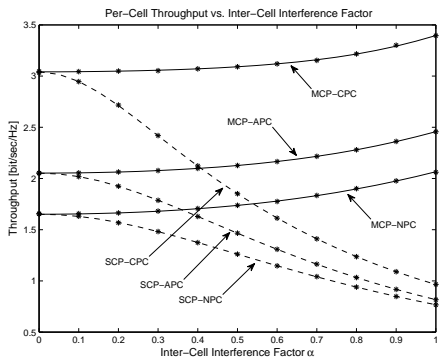} \vspace{-0.45cm}
\end{center}
\caption{Per-cell throughput vs. the inter-cell interference factor $\protect%
\alpha$ supported by the SHO model ($\protect\alpha_0=1,\ \protect\alpha_1=%
\protect\alpha$), for non-activity probability $q=0.3$, and total cell
power $P=5\ [\mathrm{dB}]$.}
\label{fig: Rate vs al SHO}
\end{figure}


In this section, we present numerical results for various settings of
interest. For all settings, results of the three power control schemes (NPC,
APC and CPC) for both MCP and SCP approaches are presented. Unsurprisingly,
for each power control scheme it is observed that joint MCP is always
beneficial over SCP. Moreover, for MCP, CPC is beneficial over APC, which in
turn is beneficial over NPC. This is because for MCP the resulting channel,
given an activity pattern, is a MAC channel whose sum-rate increases without
bound with the total transmit power. The same relations are also observed
for SCP under all tested conditions. We further note that all the rates
presented here are plotted using both analytical expressions (continuous
and dashed lines for MCP and SCP respectively) and Monte-Carlo (MC)
simulations (marked by asterisks). Examining Figures \ref{fig: Rate vs power}%
-\ref{fig: Rate vs al SHO} an excellent match between the MC and exact
results is observed for all cases over a wide range of system parameters.

In Figure \ref{fig: Rate vs power}, the per-cell throughputs supported by
the SHO model ($\alpha _{0}=1,\alpha _{1}=0.5$) are plotted as functions of
the total cell power $P$, for $K=5$ users per-cell, and non-activity
probability $q=0.3$. In addition to the notable power offset gain of the
MCP-CPC rate over MCP-NPC and -APC rates, an interference limited behavior
is observed for the rates of all SCP schemes, while those of all the MCP
schemes increase without bound with the power $P$. The throughputs \textit{%
per-active user}\footnote{%
Due to the SLLN this amounts to dividing the per-cell rate by $(1-q)K$.} of
the same setting but with $P=5\ [\mathrm{dB}]$ are plotted as functions of the
non-activity probability $q$. It is observed that the per-active user rates
of all schemes increase with $q$. Obviously, the MCP rates coincide for $%
q=0$. (The same applies to the SCP rates.) In Figure \ref{fig: Rate vs K} we
demonstrate the important role played by the number of users per-cell $K$
for the same setting of Figure \ref{fig: Rate vs power}. Examining the
figure, the benefit of MCP-CPC over all other schemes is observed. This is
because its resulting average power increases without bound with $K$ (under
fixed total cell power), while the other MCP schemes result in a bounded
power and all SCP schemes yield bounded SINRs. 
Finally, in Fig. \ref{fig: Rate vs al SHO} the per-cell throughputs
supported by the SHO model ($\alpha _{0}=1,\alpha _{1}=\alpha $) for $q=0.3$%
, and total cell power $P=5\ [\mathrm{dB}]$, are plotted as functions of the
inter-cell interference factor $\alpha $. Obviously, for each power scheme
the rates of MCP and SCP coincide when no inter-cell interference is present 
($\alpha =0$). Also notable is that all SCP rates decrease with increasing $%
\alpha $. This is because the inter-cell interference power increases with $%
\alpha $ while the useful signal power is unchanged. In contrast, the MCP
rates increase with $\alpha $ for the SHO model. It is noted that the latter
does not hold in general when interferences stem from more than one cell. 
\vspace{-0.0cm}

\section{Concluding Remarks}

\vspace{-0.0cm}
In this work we have studied the performance of a non-fading
cellular uplink system with cooperative BSs and random numbers of users per
cell. Using the simple Wyner model family and focusing on the large system
limit, we have established an analogy between the per-cell throughput of the
dynamic cellular uplink and the achievable rate of an LTI ISI channel with
flat fading, under certain conditions. In particular we have shown that the
power control scheme being used in the cellular uplink determines the fading
statistics of the ISI channel. Using a recent result regarding the
achievable rate of the ISI channel \cite{Tulino-Verdu-Caire-Shamai-ISIT08},
expressions for the cellular throughput are provided. Moreover, for cases
where interference stems from only one adjacent cell, the rate is explicitly
given, revealing analytically the impact of different system parameters.
Finally, we have demonstrated the benefits of joint MCP over the
conventional SCP approach for several cases of interest. For instance, we
have shown that combining cooperative BSs and cognitive MTs provides a dramatic
increase in system performance.


\vspace{-0.0cm}

\section*{Acknowledgment}

This work was supported by a Marie Curie Outgoing International Fellowship
and the NEWCOM++ network of excellence both within the 6th and 7th EU Framework Programmes, the U.S. National Science Foundation under Grants
CNS-06-26611, CNS-06-25637 and CCF-09-14899, and the REMON consortium for wireless communication.

\vspace{-0.0cm}
\bibliographystyle{ieeetr}
\bibliography{Reference_List}

\end{document}